%% file: TransProbDiffChannel.tex
\documentclass[10pt,twoside]{IEEEtran}
\pdfoutput=1  
\usepackage[T1]{fontenc}
\usepackage[latin9]{inputenc}
\usepackage{epsfig,psfrag}
\usepackage{graphicx}
\usepackage{amsmath,amsfonts,amssymb,amsxtra,bm,amsthm}
\usepackage{color}
\usepackage{paralist, tabularx}
\usepackage{subfigure}
\usepackage{balance}
\usepackage[adjust]{cite}
\usepackage{enumitem}
\usepackage{xfrac}
\usepackage[disable]{todonotes}
\usepackage{multirow}
\usepackage{siunitx}
\usepackage{tikz}

\newcommand{\ve}[1]{\mathbf{#1}}

\newcommand{\rmv}{\hspace*{-.3mm}}

\newtheorem{Theorem}{Theorem}

\newtheorem{Example}{Example}
\newcommand{\revMath}[1]{\color{black}{#1}\color{black}}  

\graphicspath{{./Fig/}{../../Matlab/Fig/}}
\DeclareGraphicsExtensions{.pdf}
\begin{document}

\title{On the Impact of Transposition Errors in \\ Diffusion-Based Channels}

\author{Werner~Haselmayr,~\IEEEmembership{Member,~IEEE,}
        Neeraj Varshney,~\IEEEmembership{Student Member,~IEEE,}
        \mbox{A. Taufiq~Asyhari},~\IEEEmembership{Member,~IEEE,} \\
        Andreas~Springer,~\IEEEmembership{Member,~IEEE,}
        and~Weisi~Guo\IEEEauthorrefmark{1},~\IEEEmembership{Member,~IEEE}

\thanks{Manuscript received month day, 2018; revised month day, year; accepted month day, year.  \IEEEauthorrefmark{1}Corresponding author. 
W. Haselmayr and A.~Springer are with the Johannes Kepler University Linz, Austria (\mbox{email: \{werner.haselmayr, andreas.springer\}@jku.at}). 
N. Varshney is with the Department of Electrical Engineering, Indian Institute of Technology Kanpur, India (\mbox{email: neerajv@iitk.ac.in}).
A. T. Asyhari is with the Centre for Electronic Warfare, Information and Cyber, Cranfield University, United Kingdom (\mbox{email: taufiq-a@ieee.org}). 
 W. Guo is with the School of Engineering, University of Warwick, United Kingdom (\mbox{email: weisi.guo@warwick.ac.uk)}.}}

\maketitle

\input{Abstract}

\section{Introduction}
\label{sec:intro}
\input{Introduction}

\section{System Model}
\label{sec:sys_mod}
\input{SystemModel}

\section{Bit Error Probability Analysis}
\label{sec:trans_effect}
\input{TranspositionEffect}


\section{Simulation Results}
\label{sec:sim_res}
\input{SimulationResults}

\section{Conclusions}
\label{sec:concl}
\input{Conclusions}

\ifdefined\ACK
  \section{Acknowledgment}
  \input{Acknowledgment}
\fi


\bibliographystyle{IEEEtran}
\bibliography{IEEEabrv,References}

\end{document}

%% file: Abstract.tex

\begin{abstract} 
  In this work, we consider diffusion-based molecular communication with and without drift between two static nano-machines. We employ type-based information encoding, releasing a single molecule per information bit. 
  At the receiver, we consider an asynchronous detection algorithm which exploits the arrival order of the molecules. In such systems, transposition errors fundamentally undermine reliability and capacity. Thus, in this work 
  we study the impact of transpositions on the system performance. Towards this, we present an analytical  expression for the exact bit error probability~(BEP) caused by transpositions and derive computationally tractable approximations of the BEP for diffusion-based channels with and without drift. Based on these results, we analyze the BEP when background is not negligible and derive the optimal bit interval that minimizes the BEP. Simulation results confirm the theoretical results and show the error and goodput performance for different parameters such as block size or noise generation rate.

\end{abstract}
 
\begin{IEEEkeywords}
  Asynchronous detection, diffusion-based channels, L\'{e}vy distribution, molecular communication, inverse Gaussian distribution, transposition effect.
\end{IEEEkeywords}

%% file: Introduction.tex

\IEEEPARstart{M}{olecular} communication (MC) broadly defines the transmission of information using biochemical molecules over multiple distance scales \cite{Nakano2013,Farsad16}. Within multi-cellular organisms, MC within cells, between local cells, and across the body of the organism (e.g., hormones) is essential for coordinated cellular action-reaction. Between organisms, MC takes place over several kilometers distance in air and under water (e.g., pheromones), and is used to signal intent, assist navigation, and warn of impending dangers~\cite{Wyatt2003}. The aforementioned MC largely relies on messenger molecules to transverse channels using some form of normal or anomalous diffusion mechanism, potentially combining microscopic discrete random walk with macroscopic continuum fluid mechanics. 

Due to the potential for ultra-high energy efficiency \cite{Rose15ICC}, device dimension scalability, and bio-compatibility, diffusion-based MC has gathered intense research interest in recent time. \revMath{Currently, the majority of research in diffusion-based MC can be split between\footnote{Although most research activities fall in one of these categories, we are aware that there exists other areas where significant efforts have been made.}: {(i)~Fundamental understanding and modeling of molecular signaling (e.g., \cite{Pierobon11, Srinivas12,Yilmaz14})}; {(ii)~Design, fabrication and testing of human-made MC systems (e.g.,~\cite{Koo16,Farsad17,Diez17})}; {(iii)~Applying the MC paradigm to nano-medicine applications (e.g.,~\cite{Chahibi17, Okonkwo17,Felicetti16}).}
}

\subsection{Motivation and Related Work} 
In diffusion-based MC the information can be encoded in the molecular concentration level, the release time of the molecules, and the type of molecules \cite{Farsad16}. Moreover, a combination of the aforementioned techniques is  also  possible.  Most existing work in MC considers information encoding in the molecular concentration level (e.g., \cite{Mahfuz16} and the references therein). The detection algorithms are based on the received concentration level, which is sampled at  pre-determined time instants. The detector can rely on the law of large numbers, whereby the arrival time of the peak does not vary significantly. For concentration-encoded MC intersymbol interference (ISI) is the dominant error source and a vast amount of work  has been devoted to this issue in the past (e.g., \cite{Arjmandi17} and the references therein).

However, since it is envisioned that MC employs nano-machines with very limited capabilities, it is very likely that molecular signals are represented by a limited set of molecules or molecular clusters rather than on the emission of a large number of molecules \cite{Lin15}. Here, the detection algorithms for time- and type-based information encoding exploit the arrival time or the arrival order of the molecules, respectively. Due to the stochastic nature of diffusion-based channels, it may occur that a sequence of transmitted molecules  arrive out of order at the receiver, i.e. molecules that are released earlier arrive \mbox{late~--~yielding} so-called transpositions\footnote{In \cite{Yeh12}, transpositions are referred to as crossovers.} of bits or symbols~\cite{Yeh12}. Thus, for time- and type-based information encoding using individual molecules transposition errors are the dominant error source. 
The implementation of an optimal maximum likelihood (ML) detector is almost impractical, even for a short sequence of molecules, since all possible permutations must be taken into account \cite{Srinivas12,Murin17}. 
For time-based information encoding a sub-optimal detector is proposed in~\cite{Murin17}, which cannot be adopted for type-based information encoding.
For type-based information encoding the channel capacity is derived in \cite{Hsieh_13}, assuming only transpositions between neighboring bits. In~\cite{Haselmayr17,Ahmadzadeh17} the impact of transposition errors for diffusion-based MC with mobile transmit and receive nano-machines is investigated. Different techniques for mitigating transposition errors are considered \mbox{in~\cite{Lin15,Nakano10,Manocha16,Ko12,Shih13,Weisi16_1}}. These approaches  can be divided into three categories: Sender-oriented, environment-oriented, and receiver-oriented. The works in \mbox{\cite{Nakano10,Ko12,Shih13,Weisi16_1,Lin15}} consider sender-oriented techniques for combating transpositions. Obviously, 
increasing the bit interval duration lowers the probability of transpositions. \mbox{In~\cite{Lin15}, releasing} multiple molecules per information bit, instead of a single molecule, is presented as a promising approach. Various block coding techniques are proposed in~\mbox{\cite{Ko12,Shih13,Weisi16_1}}. The code design is no longer based on the Hamming distance, which is only useful if the bits are corrupted by noise and is likely to be ineffective in the case of transposition errors. Thus, other attributes, such as for example the Hamming weight~\cite{Weisi16_1} or a molecular coding distance \cite{Ko12} are considered as suitable coding design paradigms. Environment-oriented approaches for transposition error mitigation are presented in \cite{Manocha16,Nakano10}. In \cite{Nakano10} different propagation mechanisms are investigated (diffusion, diffusion with amplification, diffusion in fixed volume space, motor-driven diffusion) and in \cite{Manocha16} a Dielectrophoresis-based relay system is proposed to maintain in-sequence delivery of molecules. Such a system converts the random diffusion into a controlled and guided drift by collecting molecules on electrodes and relaying them at a controlled interval. In \cite{Nakano10} buffering at the receiver is proposed as a receiver-oriented approach for combating transpositions.

Although the works mentioned above consider different aspects of the transposition effect a comprehensive performance analysis, especially for diffusion-based channels without drift, is lacking in the current literature.

\subsection{Contributions}
In this work, we investigate the impact of transpositions on the performance of diffusion-based MC systems with and without drift. We employ type-based information encoding, releasing a single molecule per information bit and the asynchronous detection algorithm exploits the arrival order of the molecules. We present an analytical expression for the exact bit error probability (BEP), which takes all possible permutations into account. Since evaluating the exact BEP expression is only feasible for short sequences, we also derive computationally tractable approximations of the BEP. For pure diffusion channels, we derive a tight upper bound based on the probability of out of order sequence delivery and the average BEP over these permutations. For drift channels, we approximate the BEP based on the observation that transpositions between neighboring bits are dominant. To the best of our knowledge, neither an exact BEP nor a tight bound for pure diffusion channels has been reported in the existing literature. For drift channels, we obtain the same results as in \cite{Lin15}, but we applied an alternative derivation, revealing some new interesting insights. Additionally, we provide a theoretical BEP analysis when background noise is not negligible and derive the optimal bit interval that minimizes the BEP. Through numerical results we show the error and goodput performance with and without background noise for different parameters and confirm the theoretical results. It is important to note that the presented study is important for all applications which require in-sequence delivery of single molecules, e.g., diagnostic (in-vitro) \cite{Gascoyne02} and drug assessment systems \cite{Kang08}.


\subsection{Organization}
The rest of the paper is organized as follows: Section~\ref{sec:sys_mod} presents the system model and the considered propagation environment. In Section~\ref{sec:trans_effect} we first derive the exact and approximate BEP caused by transpositions and then we investigate the BEP when background noise is not negligible. In Section \ref{sec:sim_res} we show the error and goodput performance with and without background noise through numerical results. Finally, Section~\ref{sec:concl} provides concluding remarks.


%% file: SystemModel.tex

We consider a semi-infinite one-dimensional (1D) fluid environment, whereby the length of propagation is large compared to width dimensions (e.g., blood vessels).
We assume constant temperature $T_\text{a}$ and viscosity~$\eta$. A point transmitter (TX) and a point receiver (RX) are placed at a distance $d$. 
We consider the transmission of $K$ information bits $\ve{b} = [b_1 ,...,b_K]$, where $b_k \in \{0,1\}$ denotes the transmitted bit in the $k$th bit interval. We employ 
binary molecule shift keying \cite{Farsad16}, which maps bit~0 or bit~1 to a single molecule of type-$a$ or type-$b$, respectively. Both molecule types have the same diffusion coefficient $D$. The molecules $\ve{m} = [m_1 ,...,m_K]$, with $m_k \in \{a,b\}$, are released at time
\begin{align}
  X_k = (k-1)T, \quad k = 1,\ldots, K,
  \label{eq:release_time_k}
\end{align}

where $T$ denotes the duration of the bit interval. Each molecule propagates independently from others in the environment based on a specific propagation mechanism (e.g., Brownian motion with positive drift). Similar to \cite{Nakano10}, we assume no collisions among the information molecules. At the RX, a fully absorbing receiver detects the type of the molecule and removes it from the environment. The arrival time of a molecule released at time $X_k$ is given by
\begin{align}
  Y_k = X_k + Z_k, \quad k = 1, \ldots, K,
  \label{eq:arrive_time_k}
\end{align}
where $Z_k$ denotes the random propagation time of a molecule until the first arrival, which is referred to as first hitting time. The received molecules  \mbox{$\hat{\ve{m}} = [\hat{m}_1,...,\hat{m}_K]$} are collected based on their arrival order; for example, two released molecules $m_1$ and $m_2$ that are received out of order, i.e. \mbox{$Y_1 > Y_2$},  results in $\hat{\ve{m}} = [m_2,m_1]$. Hence, we consider an asynchronous detection, which requires no synchronization between TX and RX \cite{Hsieh_13}. Finally, the received molecules are mapped to the estimated bit sequence \mbox{$\ve{\hat{b}} = [\hat{b}_1 ,...,\hat{b}_K]$}. 


\subsection{Propagation Environment}
\label{subsec:prop_env}

We consider two mechanisms by which the molecules propagate from the TX to the RX inside the fluid medium: 

\subsubsection{Brownian motion without drift}
The first hitting time of diffusion-based channels without drift follows a L\'{e}vy distribution, with its probability density function (PDF) given by~\cite{Farsad15}
\begin{align}
  f_Z(z) & = \sqrt{\frac{c}{2\pi z^3}}\exp\left(-\frac{c}{2z}\right),  \quad z > 0,
  \label{eq:levy_pdf}
\end{align}

and its cumulative distribution function (CDF) can be expressed as
\begin{align}
  F_Z(z) & = \text{erfc}\left(\sqrt{\frac{c}{2z}}\right),  \quad z > 0,
  \label{eq:levy_cdf}
\end{align}

with the scale parameter \mbox{$c=d^2/(2D)$}. The complementary error function $\text{erfc}(x)$ is defined by \mbox{$\text{erfc}(x) = 2/\sqrt{\pi}\int_x^\infty \exp(-t^2) \text{d}t$} and 
$d$ denotes the distance between TX and RX. The diffusion coefficient of the released molecules is given by

\begin{align}
  D = \frac{k_\text{B} T_\text{a}}{6 \pi \eta r},
  \label{eq:diff_coeff_def}
\end{align}
where $k_\text{B} = 1.38\times 10^{-23}\,\text{J/K}$ corresponds to the Boltzmann constant and $r$ denotes the radius of the molecules, respectively.

\subsubsection{Brownian motion with drift}
The first hitting time of diffusion-based channels with positive drift follows an inverse Gaussian distribution with its PDF given by~\cite{Srinivas12}
\begin{align}
  f_Z(z) & = \sqrt{\frac{\lambda}{2 \pi z^3}}\exp\left(-\lambda \frac{(z-\mu)^2}{2\mu^2 z}\right),  \quad z > 0,
  \label{eq:ig_pdf}
\end{align}

and its CDF can be expressed as
\begin{align}
  F_Z (z) & =  \phi\left(\sqrt{\frac{\lambda}{z}} \left(\frac{z}{\mu}\rmv - \rmv 1\right)\right)  \nonumber \\
                           & + \exp\left(\frac{2\lambda}{\mu}\right)\rmv\phi\left(-\sqrt{\frac{\lambda}{z}} \left(\frac{z}{\mu}\rmv + \rmv 1\right)\right), \quad z > 0,
  \label{eq:ig_cdf}
\end{align}

with mean~$\mu = d/v$, the shape parameter \mbox{$\lambda = d^2/(2D)$} and the CDF of the standard normal distribution  $\phi(x) = 1/\sqrt{2\pi} \int_{-\infty}^x \exp(-t^2/2)\text{d}t$. The positive drift velocity from TX to RX is denoted by~$v$. In case of no drift, i.e. $v=0\,\text{m/s}$, the inverse Gaussian distribution turns into a L\'{e}vy distribution. \newline

It is important to note that for diffusion-based channels without drift the first hitting time in a 3D environment with a spherical absorbing receiver can be modeled by a L\'{e}vy distribution with a scaling parameter \cite{Yilmaz14}. However, for diffusion-based channels with drift, to the best of our knowledge, no closed-form expression for the first hitting time probability in a 3D environment has been derived so far.

\revMath{
\subsection{Background Noise}
\label{subsec:bg_noise}
In diffusion-based MC it is very likely that the RX does not only capture molecules from the corresponding TX, but also from other sources (environment and/or other TXs). We refer to these unintended molecules as background noise. Similar to~\cite{Lin15}, we consider the following assumptions: (i) The number of unintended molecules in disjoint time intervals are independent; (ii) The number of captured unintended molecules does not favor any time instant; (iii) No two unintended molecules are captured at exactly the same time. If these conditions are fulfilled then the  number of unintended molecules captured by the RX can be modeled by a Poisson random process with rate $\lambda$. Thus, the inter-arrival time of the unintended molecules follows an exponential distribution with mean $\lambda^{-1}$.
}

%% file: TranspositionEffect.tex

Due to the random arrival time of the molecules, they may arrive out of order at the RX and, thus, transpositions occur~\cite{Yeh12}. In this section, we study 
the BEP caused by transposition errors. We present an analytical expression for the exact BEP and derive computationally tractable approximations of the BEP 
by exploiting the different channel properties. Based on these results, we analyze the BEP when background noise is not negligible.

\subsection{Exact BEP without Background Noise}
\label{subsec:exact_bep}

We divide the derivation of the exact BEP into two steps: First, we determine the probability that a particular permutation is received when a sequence of $K$ molecules is released. Then, we 
derive the BEP caused by each permutation. 

Let $\mathcal{P}_K$ represent the set of all possible permutations on $K$ released molecules.
\begin{align}
  \mathcal{P}_K = \{\pi_0, \pi_1, \ldots, \pi_M\},
\end{align}

with $M = K!$ and $\pi_0$ denotes permutation where the order of the released molecules does not change. Moreover, we denote $\pi_m(i)$ as the $i$th element of the permutation $\pi_m$~(e.g., \mbox{$\pi_0(i) = i\ \forall i$}). The probability of receiving a particular permutation~$\pi \in \mathcal{P}_K$, given that $\pi_0$ was released, can be described as 
\begin{align}
  \Pr[\pi, T|\pi_0] & = \Pr[Y_{\pi(1)} < Y_{\pi(2)} < \ldots < Y_{\pi(K)}],
  \label{eq:perm_prob_1}
\end{align}

where $Y_{\pi(i)} = X_{\pi(i)} +  Z_{\pi(i)}$, with $X_{\pi(i)} = \pi^\ast(i)T$ and $\pi^\ast(i) = (\pi(i) - 1)$. Hence, the arrival probability of a permutation $\pi \in \mathcal{P}_K$ can be written as
\begin{align}
  \Pr[\pi, T|\pi_0] 
                   \rmv = & \rmv \Pr[\pi^\ast(1)T \rmv + Z_{\pi(1)} \rmv < \rmv\pi^\ast(2)T + Z_{\pi(2)}\,, \ldots,   \nonumber \\
                        & \pi^\ast(K\rmv-\rmv 1)T \rmv+ \rmv Z_{\pi(K-1)} \rmv< \rmv\pi^\ast(K)T \rmv+ \rmv Z_{\pi(K)}].
  \label{eq:perm_prob_2}
\end{align}

Similar to \cite{Nakano10}, we assume that that the molecules propagate independently from each other and, thus, \eqref{eq:perm_prob_2} can be determined by
\begin{align}
  \Pr[\pi, T|\pi_0]  = & \int\limits_{0}^{\infty} f_Z(z_{\pi(K)}) \hspace*{-7mm} \int\limits_{0}^{z_{\pi(K)}+ \Delta \pi(K)T} \hspace*{-7mm} f_Z(z_{\pi(K-1)}) \cdots \nonumber \\
                \times & \hspace*{-7mm} \int\limits_{0}^{z_{\pi(3)}+ \Delta \pi(3)T} \hspace*{-7mm} f_Z(z_{\pi(2)}) \hspace*{-6mm} \int\limits_{0}^{z_{\pi(2)}+ \Delta \pi(2)T} \hspace*{-6mm} f_Z(z_{\pi(1)}) 
                   \text{d}z_{\pi(1)} \cdots \text{d}z_{\pi(K)},
  \label{eq:perm_prob_3}
\end{align}

with $\Delta \pi(i) = \pi(i) - \pi(i-1)$ and $f_Z(z)$ denotes the PDF of the L\'{e}vy and the inverse Gaussian distribution defined in \eqref{eq:levy_pdf} and \eqref{eq:ig_pdf}, respectively.

The BEP of a particular permutation can be derived by
\begin{align}
  \Pr[\hat{b}_k \neq b_k|\pi] = \frac{\text{disp}(\pi)}{2 K}, \quad \pi \in \mathcal{P}_K.
  \label{eq:perm_bit_err}
\end{align}

The function $\text{disp}(\cdot)$ computes the number of displacements of a particular permutation $\pi \in \mathcal{P}_K$ and is defined by
\begin{align}
  \text{disp}(\pi) =\sum_{i=1}^K\left\lceil \frac{|\pi(i) - i|}{K} \right \rceil,
  \label{eq:disp_perm}
\end{align}

where the ceiling function $\lceil x \rceil$ denotes the mapping to the smallest following integer number.

Combining \eqref{eq:perm_prob_1} and \eqref{eq:perm_bit_err} gives the  BEP due to transposition errors for the transmission of $K$ bits
\begin{align}
  P_\text{t}^\text{e} = \Pr[\hat{b}_k \neq b_k] = \sum\limits_{\pi \in \mathcal{P}_K} \Pr[b_k \neq \hat{b}_k|\pi] \Pr[\pi, T|\pi_0] .
  \label{eq:tot_bit_err}
\end{align}

However, the evaluation of \eqref{eq:tot_bit_err} is only feasible for small block sizes $K$, since $|\mathcal{P}_K|= K!$. Hence, in the next section we present computationally tractable approximations
of the BEP.

The following example shows the evaluation of the exact BEP for the transmission of $K=3$ bits.
\begin{figure}[t!]
  \begin{center}
    \includegraphics[scale = 1.1]{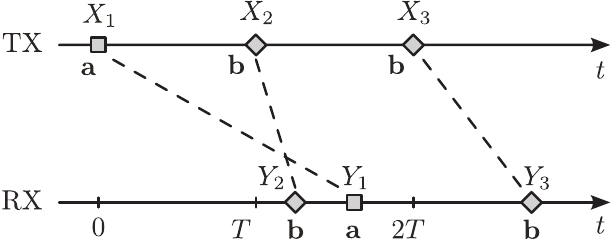}
  \end{center} 
  \caption{Illustration of the transposition effect. Three molecules of type $a$, $b$ and~$b$ are released at times $X_1$, $X_2$ and $X_3$ and arrive at times $Y_1$, $Y_2$ and~$Y_3$, with $Y_2 < Y_1 < Y_3$. Thus, the first and the second molecule are received out of order.}
  \label{fig:transp_effect}
\end{figure}

\begin{Example}
We consider the transmission of $K=3$ information bits $\ve{b} = [b_1, b_2, b_3]$ with $b_k\in\{0,1\}$. Thus, the molecules $\ve{m} = [m_1, m_2, m_3]$, with $m_k \in \{a,b\}$ are released at time $X_1 = 0$, $X_2 = T$, and $X_3 = 2T$. For $b_k = 0$ type-$a$ molecules ($m_k = a$) and for $b_k = 1$ type-$b$ molecules ($m_k = b$) are released, respectively. The released molecules propagate through the environment according to the mechanisms described in Section \ref{subsec:prop_env}. We assume that they arrive at the receive in the order $\hat{\ve{m}} = [m_2, m_1, m_3]$ as illustrated in Fig. \ref{fig:transp_effect}. For this permutation  $\pi(1) = 2$, $\pi(2) = 1$ and $\pi(3) = 3$ and  $Y_2 < Y_1 < Y_3$ holds. According to~\eqref{eq:perm_prob_3}, the probability to observe such a permutation at the receiver is given by 
\begin{align*}
  \Pr[\pi, T|\pi_0] = &\Pr[Y_{2} < Y_{1} < Y_{3}] \\
                    = & \int\limits_{0}^{\infty} f_Z(z_3) \int\limits_{0}^{z_3+2T} f_Z(z_1)  \int\limits_{0}^{z_1-T} f_Z(z_2) \text{d}z_2\text{d}z_1\text{d}z_3.
\end{align*}

As shown in Fig. \ref{fig:perm_ber}, a permutation does not directly translate to bit errors. For the  permutation $\hat{\ve{m}} = [m_2, m_1, m_3]$ no error occurs if 
the released molecules $m_1$ and $m_2$ are of the same type, i.e. $m_1=m_2 =m$ with $m \in \{a,b\}$. According to~\eqref{eq:perm_bit_err}, the BEP for the considered permutation $\pi$ is given by 
\begin{align*}
  \Pr[\hat{b}_k \neq b_k|\pi] = \frac{1}{3},
\end{align*}

with 
\begin{align*}
  \text{disp}(\pi) = \left \lceil \frac{|2 - 1|}{3} \right \rceil + \left \lceil \frac{|1 - 2|}{3} \right \rceil +  \left \lceil \frac{|3 - 3|}{3} \right \rceil = 2.
\end{align*}

\end{Example}

\begin{figure}[t!]
  \begin{center}
    \includegraphics[scale = 1]{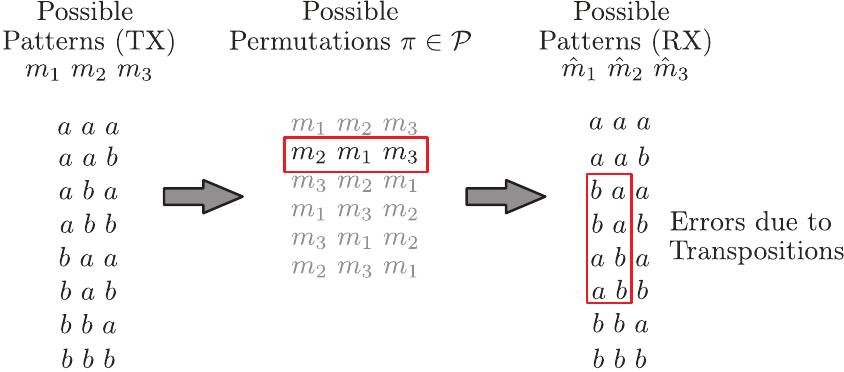}
  \end{center} 
  \caption{Relation between permutation and bit error. If the first and the second molecule arrive out of order, i.e $[m_1, m_2, m_3] \rightarrow [m_2, m_1, m_3]$, an error occurs only if the released molecules $m_1$ and $m_2$ are of different type.}
  \label{fig:perm_ber}
\end{figure}

\revMath{
\subsection{Approximate BEP without Background Noise}
\label{subsec:appr_bep}
We approximate the exact BEP given in \eqref{eq:tot_bit_err} by exploiting the different properties of pure diffusion and drift channels. For drift channels, we assume 
that transpositions between neighboring bits are dominant and transpositions between bits that are more than one bit interval apart are negligible \cite{Lin15}. This assumption holds for $T > d/v$, which means that the average time of the molecules to arrive at the RX $\mu = d/v$ is smaller than the duration of one bit interval~$T$. Unfortunately, this assumption does not hold for pure diffusion channels, since in this case the average arrival time of the molecules is infinity. For pure diffusion channels, we approximate the BEP based on the probability of out of order sequence delivery and the average BEP over these permutations.

\subsubsection{Diffusion-Based Channel Without Drift}

We define a lower bound for the arrival probability of the permutation $\pi_0$, corresponding to in-sequence delivery of the released molecules. For this, we assume that all molecules arrive within the finite observation interval $KT$ and, thus, the lower bound can be expressed as
\begin{align}
  \Pr[\pi_0, T|\pi_0]  \rmv & \geq  \rmv \Pr[Z_1 \rmv < \rmv KT, Z_2 \rmv < \rmv (K-1)T, \ldots, Z_{K} \rmv< \rmv T] \nonumber \\ 
                            & = \int\limits_{0}^{T} f_Z(z_K) \cdots \hspace*{-4mm} \int\limits_{0}^{(K-1)T} \hspace*{-3mm} f_Z(z_{2}) \int\limits_{0}^{KT} f_Z(z_{1})\text{d}z_1 \cdots \text{d}z_K\nonumber \\
                            & = \prod\limits_{j=1}^{K}F_Z(jT),
  \label{eq:in_seq_prob_upper}
\end{align}

%

where $F_Z(z)$ denotes the CDF of the L\'{e}vy distribution defined in \eqref{eq:levy_cdf}. 
Hence, the arrival probability of all permutations except~$\pi_0$, i.e. $\mathcal{P}_K\backslash\{\pi_0\}$, is given by 
\begin{align}
  \Pr[\mathcal{P}_K\backslash\{\pi_0\}, T|\pi_0] & = 1- \Pr[\pi_0, T|\pi_0] \nonumber \\ 
                                                 & \leq 1 -\prod\limits_{j=1}^{K}F_Z(jT).
  \label{eq:out_seq_prob}
\end{align}
Next, we derive the average BEP of the permutations in $\mathcal{P}_K\backslash\{\pi_0\}$. The following theorem allows to calculate the average value of $\text{disp}(\pi)$ in \eqref{eq:disp_perm} over all permutations $\pi \in \mathcal{P}_K$.

\begin{Theorem}
  Among all permutations of length $K \geq 1$ the average displacement is given by
  \begin{align}
    \overline{\text{disp}}(\pi) =  \frac{1}{K!}\sum\limits_{\pi \in \mathcal{P}_K} \text{disp}(\pi) = K-1,
    \label{eq:avg_disp_perm}
  \end{align}
with  $\text{disp}(\pi) =  \sum_{i=1}^K \lceil |\pi(i) - i|/K \rceil$ as defined in \eqref{eq:disp_perm}.
\end{Theorem}
\begin{proof}
  Let us assume to pick $i \in \{1,\ldots,K\}$. Then the number of permutations $\pi \in \mathcal{P}_K$ that map $i$ onto $i'$, i.e. $\pi(i)=i'$, is $(K-1)!$.
  Thus, we have 
  \begin{align}
    \sum\limits_{\pi \in \mathcal{P}_K} \left \lceil \frac{|\pi(i) - i|}{K} \right \rceil = (K-1)(K-1)!.
  \end{align}
  Now, the average displacement can be obtained as follows
  \begin{align}
     \overline{\text{disp}}(\pi) & =  \frac{1}{K!}\sum\limits_{\pi \in \mathcal{P}_K} \text{disp}(\pi) \nonumber \\
                                 & =  \frac{1}{K!}\sum_{i=1}^K \sum\limits_{\pi \in \mathcal{P}_K} \left\lceil \frac{|\pi(i) - i|}{K} \right \rceil \nonumber \\
                                 & = K-1,
  \end{align}
  with $\text{disp}(\pi) =  \sum_{i=1}^K \lceil |\pi(i) - i|/K \rceil$ as defined in \eqref{eq:disp_perm}.
\end{proof}

Replacing $\text{disp}(\pi)$ by $\overline{\text{disp}}(\pi)$ in \eqref{eq:perm_bit_err} gives the average BEP of the permutations in  $\mathcal{P}_K\backslash\{\pi_0\}$
\begin{align}
  \overline{\text{Pr}}[\hat{b}_k \neq b_k|\pi] = \frac{\overline{\text{disp}}(\pi)}{2K} = \frac{K-1}{2K}.
  \label{eq:avg_bep}
\end{align}

Now, we can define an upper bound for the BEP by using the arrival probability of the permutations in $\mathcal{P}_K\backslash\{\pi_0\}$ defined in~\eqref{eq:out_seq_prob} and their average BEP given in~\eqref{eq:avg_bep} 
\begin{align}
  P_\text{t}^\text{e} = \text{Pr}[\hat{b}_k \neq b_k] 
                                    & \leq \frac{K-1}{2K}\left(1-\prod\limits_{j=1}^{K}F_Z(jT)\right).
  \label{eq:bep_appr_pure}
\end{align}

In Sec. \ref{sec:sim_res}, we show through numerical results that \eqref{eq:bep_appr_pure} provides a tight upper bound, and, thus, is an appropriate BEP approximation for diffusion-based channels without drift.}

\begin{figure}[t!]
  \begin{center}
    \includegraphics[scale = 0.6]{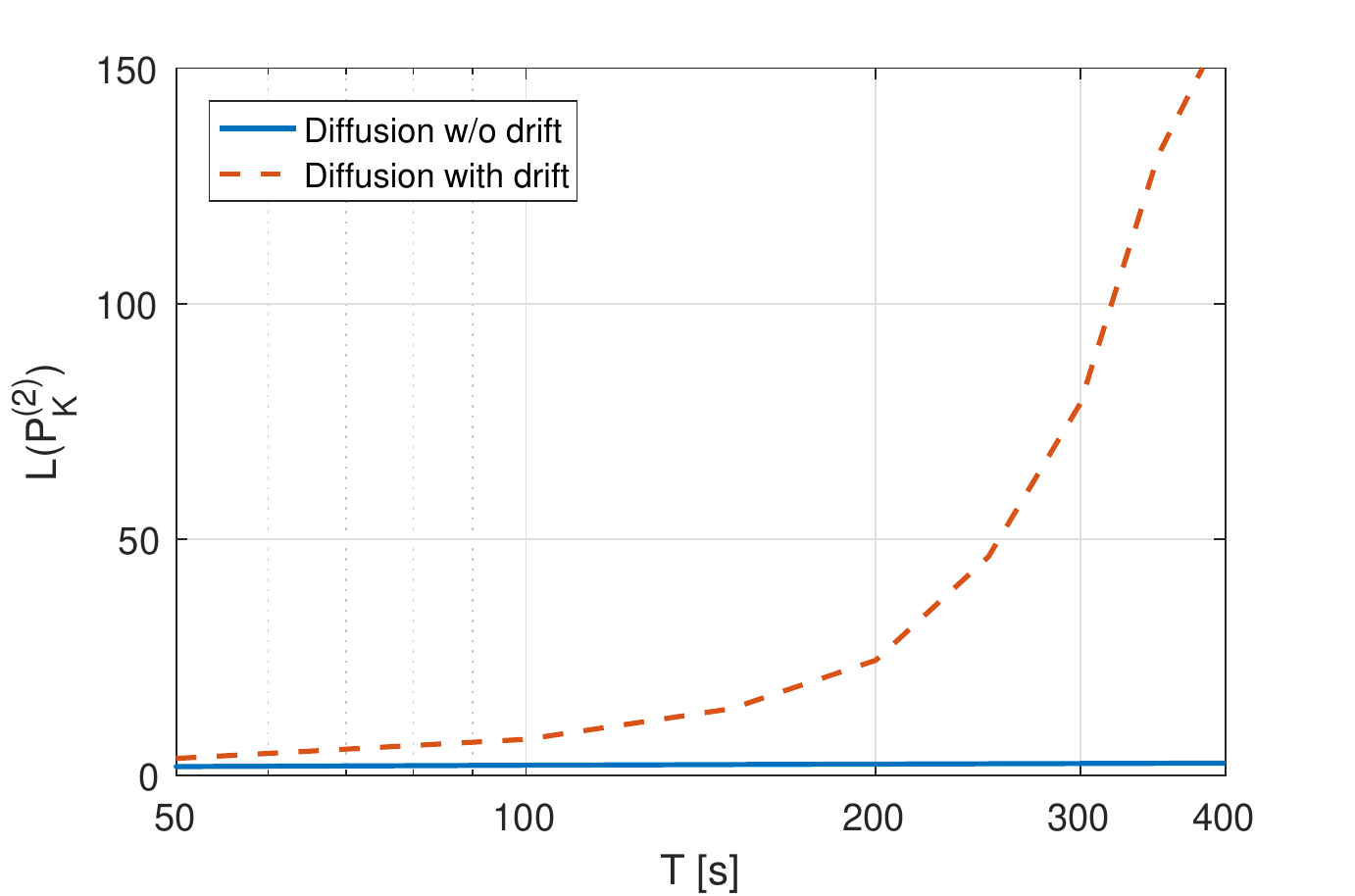}
  \end{center} 
  \caption{Ratio between the arrival probability of permutations with single and multiple transpositions as defined in \eqref{eq:ratio_arriv_prob_subset} for the simulation parameters in Tab.~\ref{tab:sim_param} with $K=4$. }
  \label{fig:PermProb_vs_T}
\end{figure}

\subsubsection{Diffusion-Based Channel With Drift}
\revMath{It can be shown, that the tight upper bound derived in the previous section also holds for drift channels. However, in this section we derive another BEP approximation, exploiting the property that in drift channels transpositions between neighboring bits are dominant. Fig. \ref{fig:PermProb_vs_T} shows the ratio between the arrival probability of permutations with a single transposition between neighboring bits and the arrival probability of permutations with multiple permutations, which can be expressed as
\begin{align}
  L(\mathcal{P}_K^{(2)}) = \frac{\Pr[\pi \in \mathcal{P}^{(2)}]}{\Pr[\pi \in \mathcal{P}_K\backslash \{\pi_0 \cup \mathcal{P}^{(2)}\}]}.
  \label{eq:ratio_arriv_prob_subset}
\end{align}

The set $\mathcal{P}_K^{(2)} \subset \mathcal{P}_K $ includes only permutations with a single transposition between neighboring bits. We observe from Fig.~\ref{fig:PermProb_vs_T} that by increasing the bit interval $T$ the ratio~$L(\mathcal{P}_K^{(2)})$ increases significantly for drift channels, but remains almost constant for pure diffusion channels. Thus, for drift channels it is very likely to receive a permutation with only a single permutation between neighboring bits instead of a permutation with multiple permutations. The probability of receiving a permutation $\pi \in \mathcal{P}_K^{(2)}$ can be expressed as\footnote{It is important to note that the probability on the right hand side of \eqref{eq:prob_single_perm} describes the arrival probability of all permutations with transposed bits $k$ and $k+1$, irrespective of the position of the other bits in the permutation sequence. However, as stated above, the arrival probability is dominated by the permutation having only the bits $k$ and  $k+1$ swapped and the other bits are received in the correct order, i.e. \mbox{$\pi \in \mathcal{P}_K^{(2)}$}.}
\begin{align}
  \Pr[\pi, T|\pi_0]  \rmv = & \Pr[Z_{k} \rmv > \rmv Z_{k+1} \rmv + \rmv T] \nonumber \\
                          = & \Pr[\Delta Z \rmv < \rmv -T], \quad k = 1,\ldots,K\rmv - \rmv 1,
  \label{eq:prob_single_perm}
\end{align}

with the random variable $\Delta Z  = Z_{k+1} - Z_{k}$,  where $Z_{k+1}$ and~$Z_{k}$ follow an inverse Gaussian distribution as defined in Sec. \ref{subsec:prop_env}. According to~\eqref{eq:perm_bit_err}, the BEP of permutations \mbox{$\pi \in \mathcal{P}^{(2)}_K$} is given by
\begin{align}
  \Pr[b_k \neq \hat{b}_k|\pi] = \frac{1}{K}, \quad \pi \in \mathcal{P}^{(2)}_K,
  \label{eq:bep_single_perm}
\end{align}

with $\text{disp}(\pi) = 2$.
Based on \eqref{eq:prob_single_perm} and \eqref{eq:bep_single_perm} the BEP for drift channels can be approximated as  follows
\begin{align}
  P_\text{t}^\text{e} = \Pr[\hat{b}_k \neq b_k] & = \sum\limits_{\pi \in \mathcal{P}_K^{(2)}} \Pr[b_k \neq \hat{b}_k|\pi] \Pr[\pi, T|\pi_0] \nonumber \\
                          & = \frac{K-1}{K}\Pr[\Delta Z \rmv < \rmv -T] \nonumber \\
                          & = \frac{K-1}{K}F_{\Delta Z}(-T),
  \label{eq:bep_appr_drift}
\end{align}

with $|\mathcal{P}_K^{(2)}|=K-1$. Moreover, $F_{\Delta Z}(\Delta z)$ denotes the CDF of the random variable $\Delta Z$ and probability \mbox{$F_{\Delta Z}(-T) = 1-F_{\Delta Z}(T)$} describes the tail probability of~$\Delta Z$. 
Unfortunately, no closed-form expressions exists for the exact distribution of~$\Delta Z$. Thus, we apply a recently proposed moment matching approximation by a the normal inverse Gaussian (NIG) distribution \cite{Haselmayr18}. The PDF of the NIG distribution is given by
\begin{align}
  f_{\Delta{Z}}(\Delta z) = & \frac{\alpha \delta}{\pi}\exp\left(\delta \sqrt{\alpha^2 - \beta^2} - \beta (\Delta z-\mu)\right) \nonumber \\
           & \times \frac{K_1\left(\alpha\sqrt{\delta^2 + (\Delta z-\mu)^2}\right)}{\sqrt{\delta^2 + (\Delta z-\mu)^2}},
  \label{eq:nig_pdf}
\end{align}

where $K_1(\cdot)$ denotes the  modified Bessel function of the third kind with index $1$ and the parameters are defined by \cite{Haselmayr18}
\begin{equation}
\begin{aligned}
  \alpha & = \frac{1}{\sqrt{20}}\frac{v^2}{D}, \quad
  &\beta &  = 0, \\
  \mu  & = 0, \quad
  & \delta &  = \frac{2}{\sqrt{5}}\frac{d}{v}.
\end{aligned}
\label{eq:nig_moments_spec1}
\end{equation}
}

\revMath{
\subsection{Approximate BEP with Background Noise}
\label{subsec:bep_bg_noise}
We approximate the BEP caused by transpositions and background noise by using the union bound
\begin{align}
  P^\text{e} \leq P_\text{t}^\text{e} + P_\text{n}^\text{e},
  \label{eq:bep_trans_bg}
\end{align}

where $P_\text{t}^\text{e}$ denotes to the approximate BEP due to transpositions (cf. \eqref{eq:bep_appr_pure} and \eqref{eq:bep_appr_drift}). The BEP caused by background noise is denoted 
by $P_\text{n}^\text{e}$ and given by \cite{Lin15}
\begin{align}
   P_{n}^\text{e} & = \frac{1}{2}\left[\gamma(1,\xi T) - \frac{1}{\xi T} \gamma(2,\xi T)\right],
\end{align}

with $\xi = 2\lambda K$ and the lower incomplete Gamma function $\gamma(s,x)$. Exploiting some properties of the incomplete Gamma function \cite{Apelblat83}, i.e. \mbox{$\gamma(1,\xi T) = 1 - \exp(-\xi T)$} and \mbox{$\gamma(2,\xi T) = \gamma(1,\xi T) - \xi T \exp(-\xi T)$}, the BEP expression due to background noise simplifies to
\begin{align}
  P_{n}^\text{e} & = \frac{1}{2}\left[1 -  \frac{1}{\xi T} +  \frac{\exp(-\xi T)}{\xi T}\right].
\end{align}

If the background noise is negligible the errors due to transpositions can be reduced by increasing the bit interval~$T$~\mbox{(cf. Sec. \ref{sec:sim_res})}. Unfortunately, if background noise is considered larger bit intervals increases 
the number of unintended molecules captured by the RX, which increases the BEP. Hence, there exists an optimal bit interval $T_\text{opt}$ that minimizes the BEP, i.e.
\begin{align}
  T_\text{opt} = \mathop{\text{argmin}}_{T} P^\text{e}(T).
\end{align}

In order to derive the optimal bit interval $T_\text{opt}$ we solve the following equation
\begin{align}
  \frac{\partial P^\text{e}}{\partial T} = \frac{\partial P_\text{t}^\text{e}}{\partial T} + \frac{\partial P_\text{n}^\text{e}}{\partial T} = 0,
\end{align}

where we assumed a tight union bound, i.e. $P^\text{e} \approx P_\text{t}^\text{e} + P_\text{n}^\text{e}$. The first derivative of $P_\text{n}^\text{e}$ can be calculated as follows
\begin{align}
  \frac{\partial P_\text{n}^\text{e}}{\partial T} = \frac{\exp(-\xi T)\left[-1 - \xi T + \exp(\xi T)\right]}{2 \xi T^2}.
  \label{eq:deriv_bg_noise}
\end{align}

For the first derivative of $P_\text{t}^\text{e}$ we use the approximate BEP expression for pure diffusion and drift channels given in~\eqref{eq:bep_appr_pure} and~\eqref{eq:bep_appr_drift}, respectively. 
For pure diffusion channels, the first derivative can be calculated by applying the generalized product rule, which results in 
\begin{align}
  \frac{\partial P_\text{t}^\text{e}}{\partial T} = - \frac{K-1}{2K} \prod\limits_{j=1}^{K}F_Z(jT) \sum\limits_{j=1}^{K} \frac{jf_Z(jT)}{F_Z(jT)},
  \label{eq:deriv_bep_drift}
\end{align}

where $f_Z(z)$ and $F_Z(z)$ denote the PDF and CDF of the L\'{e}vy distribution as defined in \eqref{eq:levy_pdf} and~\eqref{eq:levy_cdf}.
For drift channels, the first derivative is given by
\begin{align}
  \frac{\partial P_\text{t}^\text{e}}{\partial T} = - \frac{K-1}{K} f_{\Delta Z}(-T).
  \label{eq:deriv_bep_pure}
\end{align}

where $f_{\Delta Z}(\Delta z)$ corresponds to the PDF of the NIG distribution given in \eqref{eq:nig_pdf}. 
Finally, we obtain the optimal bit interval~$T_\text{opt}$ by numerically solving 
\begin{align}
  \frac{\partial P^\text{e}(T)}{\partial T}\Big|_{T=T_\text{opt}} = 0,
  \label{eq:opt_bit_interval}
\end{align}

using \eqref{eq:deriv_bg_noise} -- \eqref{eq:deriv_bep_pure}.
}

%% file: SimulationResults.tex

\renewcommand{\arraystretch}{1.2}
\begin{table}[t!]
  \caption{Simulation Parameters}
  \centering
  \begin{tabular}{|c|c|c|}
    \hline
    Parameter  &  Value \\
    \hline
    \hline
    Drift velocity $v$  &  $\{0,1\}\,\mu \text{m/s}$  \\
    \hline
    Block size $K$  &  $20$ \\
    \hline
    Distance $d$  &  $10\,\mu \text{m}$ \\
    \hline
    Diff. coeff $D$  &  $28.37\,\mu\text{m}^2\text{/s}$ \\
    \hline 
    Noise Generation Rate  $\lambda$ &  $1\times 10^{-6}\, s^{-1}$ \\
    \hline
  \end{tabular}
  \label{tab:sim_param}
  \vspace{-2ex}
\end{table}
\renewcommand{\arraystretch}{1.2}

In this section, we investigate the error and goodput performance under different operating conditions. \revMath{The goodput $G$ represents the number of successfully received bits per block and can be expressed as

\begin{align}
  G = (1-P^\text{e})K,
  \label{eq:goodput}
\end{align}
where $P^\text{e}$ denotes the BEP caused by transpositions and background noise as defined in \eqref{eq:bep_trans_bg}. The probability for a successful bit reception is given by $(1-P^\text{e})$ and $K$ is the block size. We also show the accuracy of the BEP approximations derived in Secs. \ref{subsec:appr_bep} and \ref{subsec:bep_bg_noise}. In the simulations, we considered a one-shot communication \cite{Murin17_2}, which means that the TX sends a block of $K$ information bits to the RX and then remains silent for a long period\footnote{For example, a nano-sensor that infrequently sends data to a RX.}. In particular, we assumed that all released molecules of a block eventually arrive at the RX (infinte lifetime) and the TX remains silent until all molecules are captured by the RX. However, we consider an finite
observation duration at the RX, i.e. the RX collects molecules up to time $KT$. Each data point in the simulations was obtained by generating $N=10^6$ uniformly distributed bits which are split up into $K$ blocks that are sent independently of each other to the RX (one-shot communication). Thus, no transpositions across the blocks occur.}\ If not otherwise stated, we used the simulation parameters in Tab.~\ref{tab:sim_param}. The diffusion coefficient $D = 28.37\,\mu\text{m}^2\text{/s}$ is obtained by \eqref{eq:diff_coeff_def}, substituting $T_a = 310\,\text{K}$ (body temperatur), venous blood viscosity $\eta = 0.004\, \text{Pa\,$\cdot$\,s}$ and molecule radius $r = 2\,\text{nm}$ \cite{Lin15}. We assumed that the radius of the molecules is negligible  compared to the distance between TX and RX, i.e. $r\rmv \ll\rmv d$. 
In all figures\footnote{For the interpretation of the references to the color in all figures, the reader is referred to the web version of this article.}, the blue curves indicate the BER results for pure diffusion channels ($v=0\,\mu \text{m/s}$) and the red curves show the BER results for drift channels ($v=1\,\mu \text{m/s}$). Moreover, the circle and star markers indicate the theoretical BEP approximation for pure diffusion and drift channels derived in Secs. \ref{subsec:appr_bep} and \ref{subsec:bep_bg_noise}.

\begin{figure}[t!]
  \begin{center}
    \includegraphics[scale = 0.7]{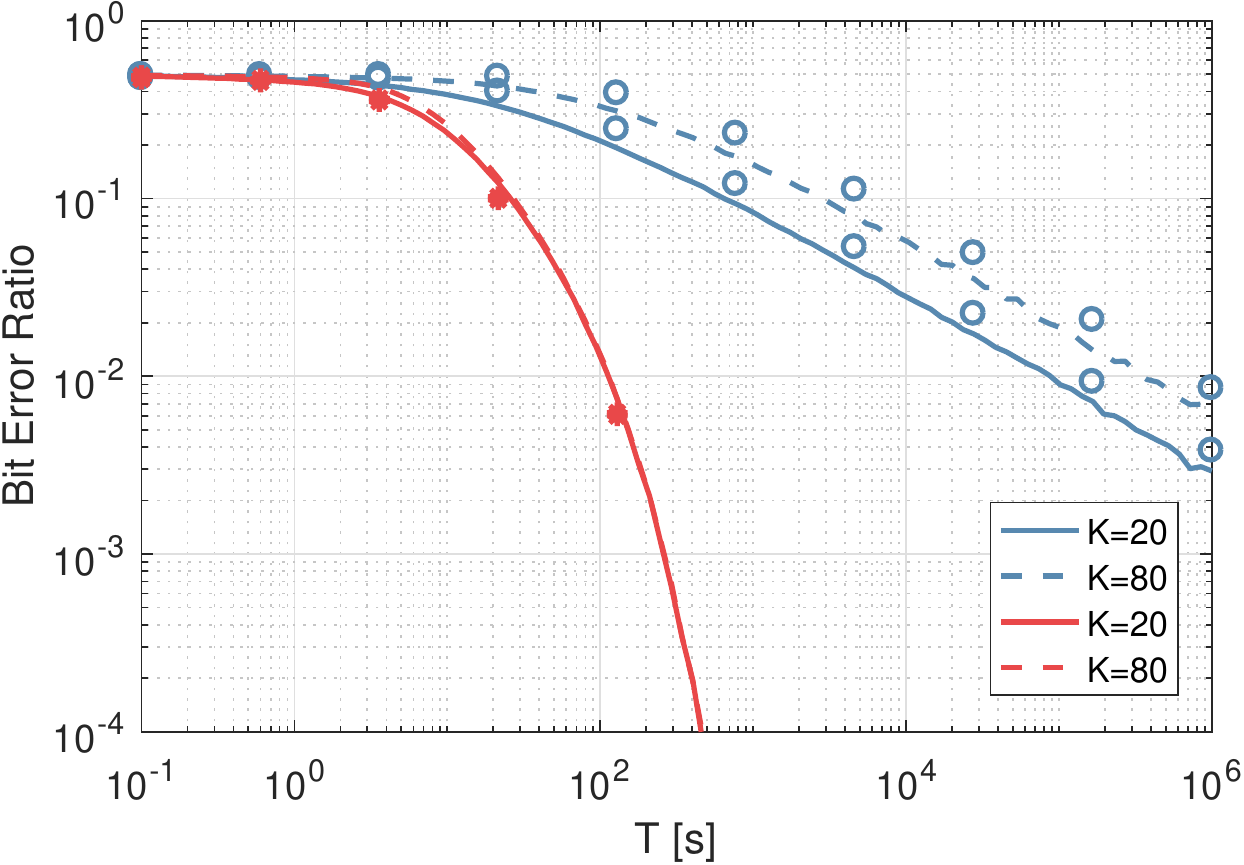}
  \end{center} 
  \vspace{-2ex}
  \caption{BER versus bit interval $T$ for different block sizes~$K$ without background noise. Blue and red curves: Simulated BER for pure diffusion and drift channels; Circle and star marker: Theoretical BEP approximation for pure diffusion and drift channels as discussed in Sec. \ref{subsec:appr_bep}.}
  \label{fig:BER_vs_T_K}
\end{figure}

\begin{figure}[t!]
  \begin{center}
    \includegraphics[scale = 0.7]{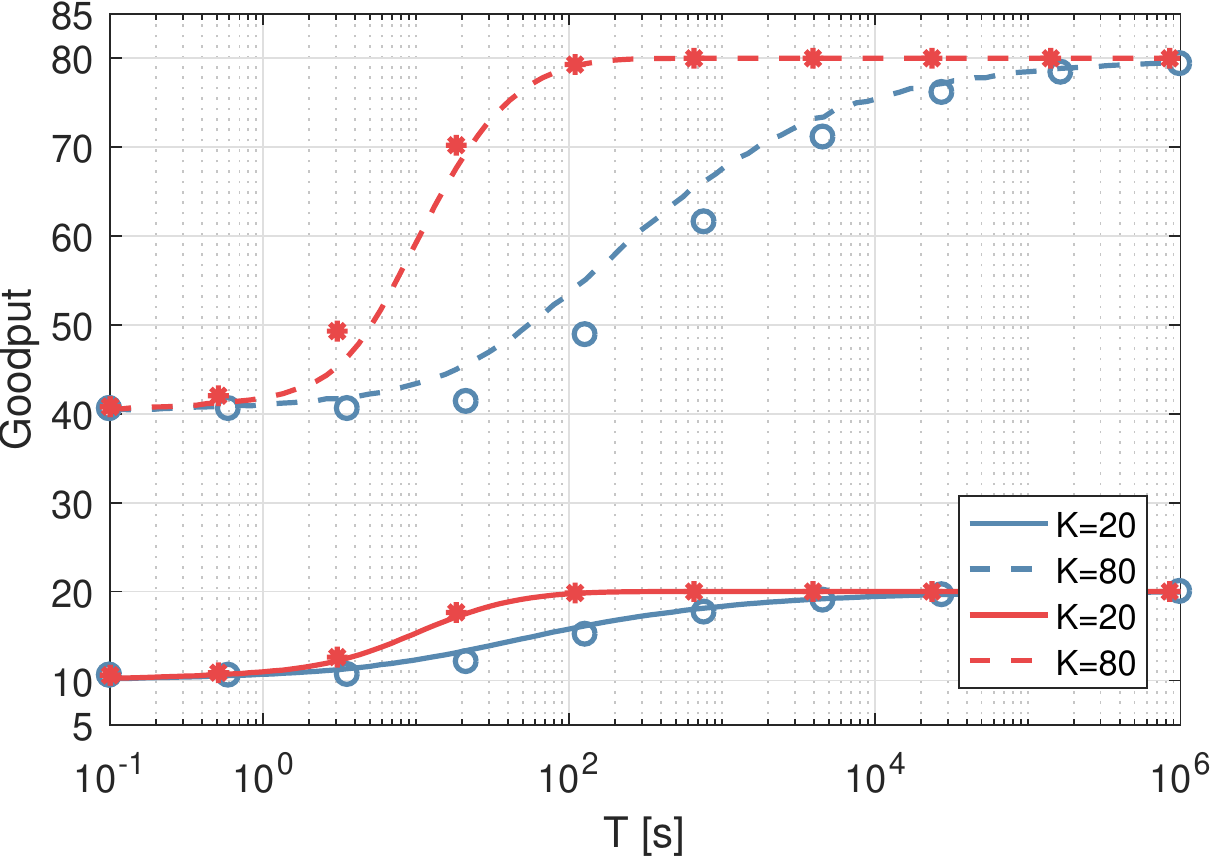}
  \end{center} 
  \vspace{-2ex}
  \caption{Goodput versus bit interval $T$ for different block sizes~$K$ without background noise. Blue and red curves: Simulated goodput for pure diffusion and drift channels; Circle and star marker: Theoretical goodput approximation for pure diffusion and drift channels as discussed in Sec. \ref{sec:sim_res}.}
  \label{fig:GP_vs_T_K}
\end{figure}

\begin{figure}[t!]
  \begin{center}
    \includegraphics[scale = 0.7]{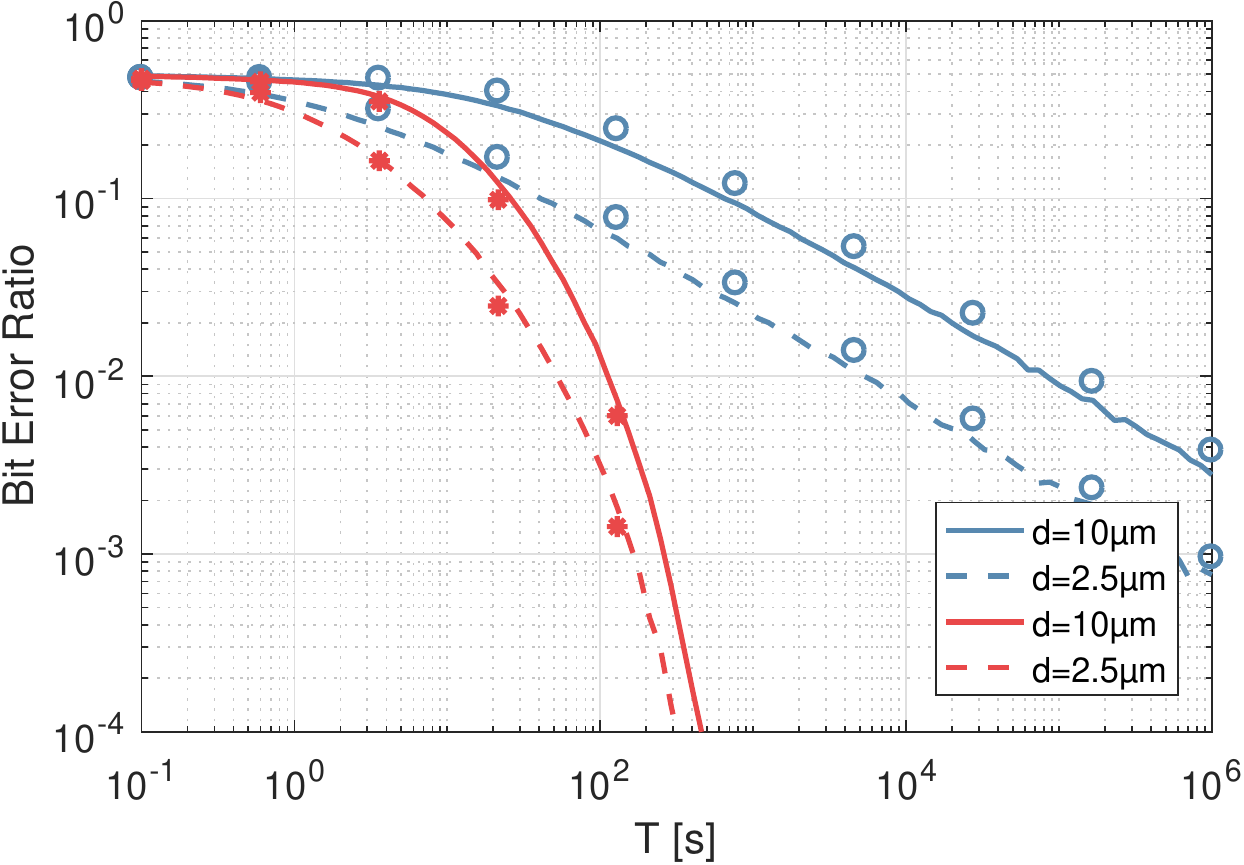}
  \end{center} 
  \vspace{-2ex}
  \caption{BER versus bit interval $T$ for different distances $d$ without background noise. Blue and red curves: Simulated BER for pure diffusion and drift channels; Circle/star marker: Theoretical BEP approximation for pure diffusion and drift channels as discussed in Sec. \ref{subsec:appr_bep}.}
  \label{fig:BER_vs_T_d}
\end{figure}

\begin{figure}[t!]
  \begin{center}
    \includegraphics[scale = 0.7]{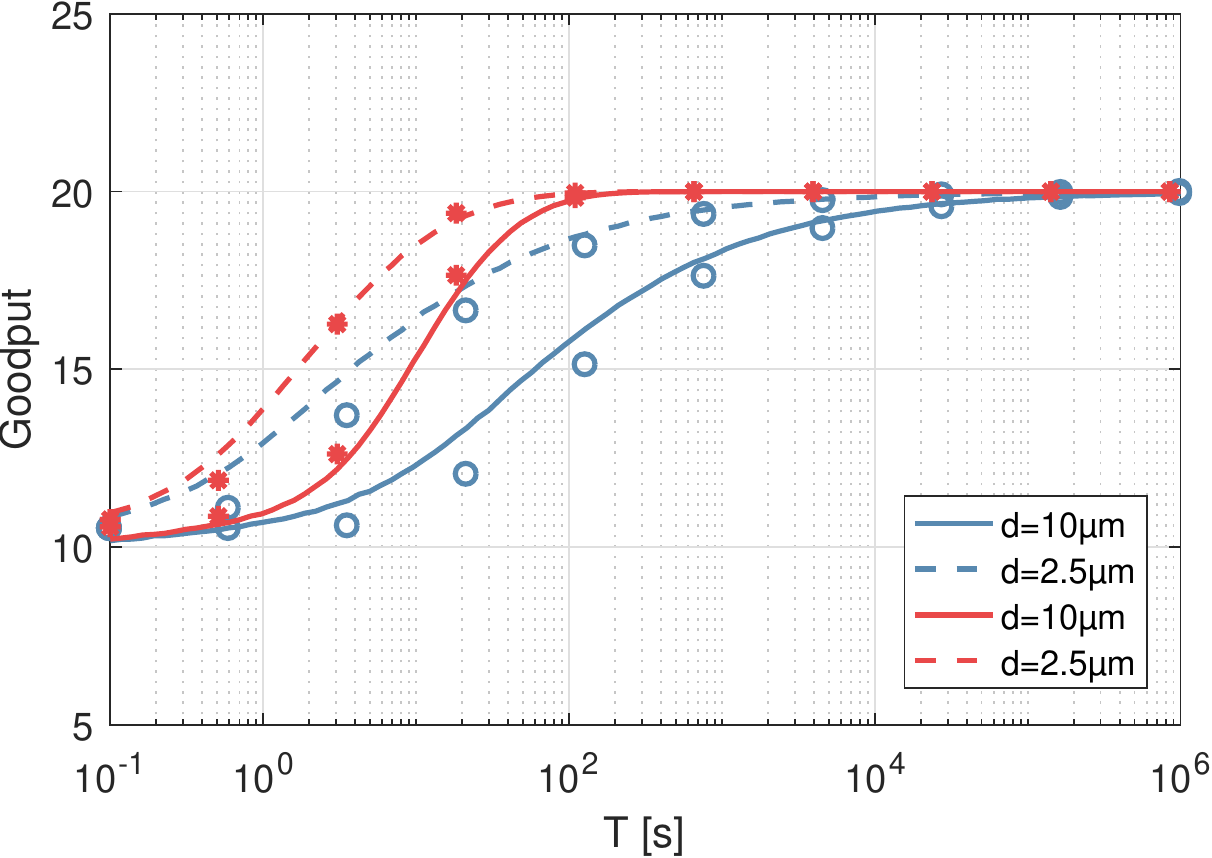}
  \end{center} 
  \vspace{-2ex}
  \caption{Goodput versus bit interval $T$ for different distances $d$ without background noise. Blue and red curves: Simulated goodput for pure diffusion and drift channels; Circle and star marker: Theoretical goodput approximation for pure diffusion and drift channels as discussed in Sec. \ref{sec:sim_res}.}
  \label{fig:GP_vs_T_d}
\end{figure}

\begin{figure}[t!]
  \begin{center}
    \includegraphics[scale = 0.7]{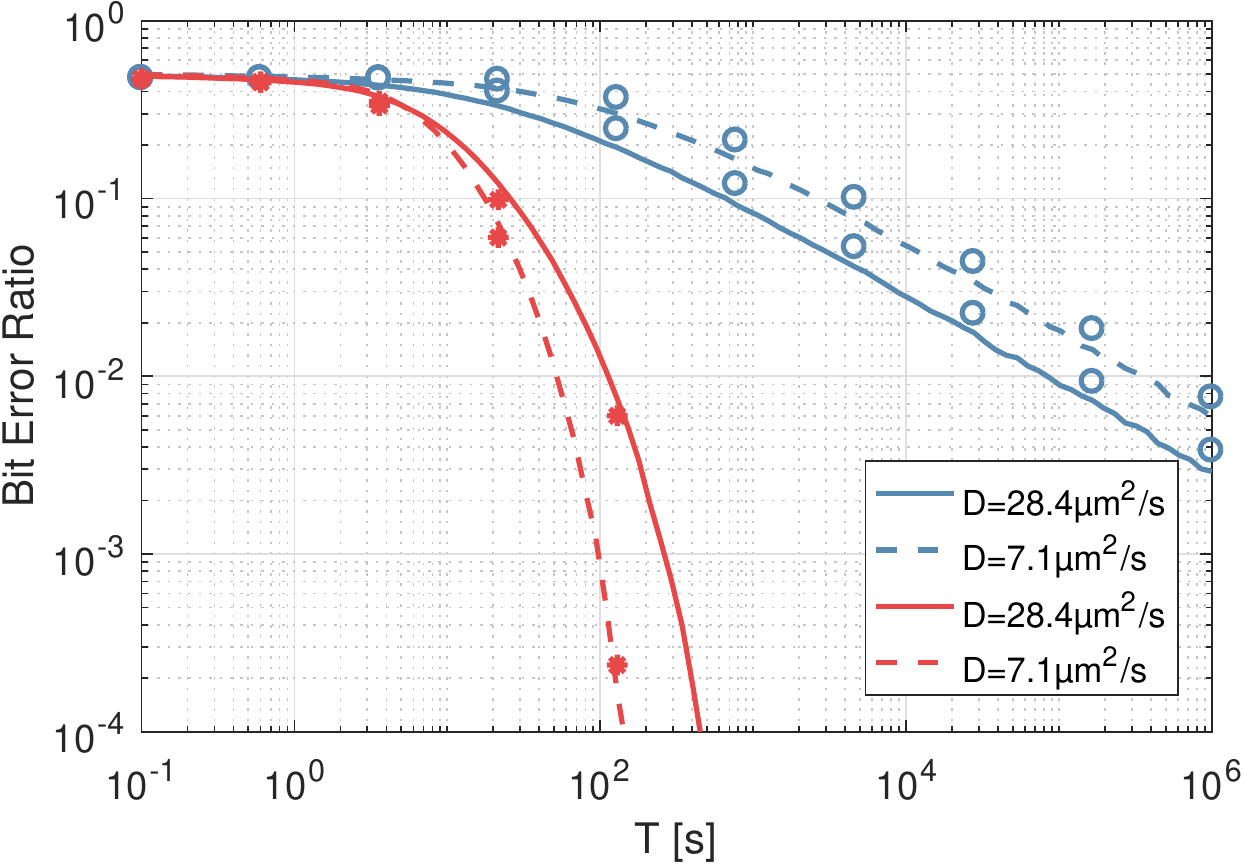}
  \end{center}
  \vspace{-2ex}
  \caption{BER versus bit interval $T$ for different diffusion coefficients $D$ without background noise. Blue and red curves: Simulated BER for pure diffusion and drift channels; Circle and star marker: Theoretical BEP approximation for pure diffusion and drift channels as discussed in Sec. \ref{subsec:appr_bep}.}
  \label{fig:BER_vs_T_D}
\end{figure}

\begin{figure}[t!]
  \begin{center}
    \includegraphics[scale = 0.7]{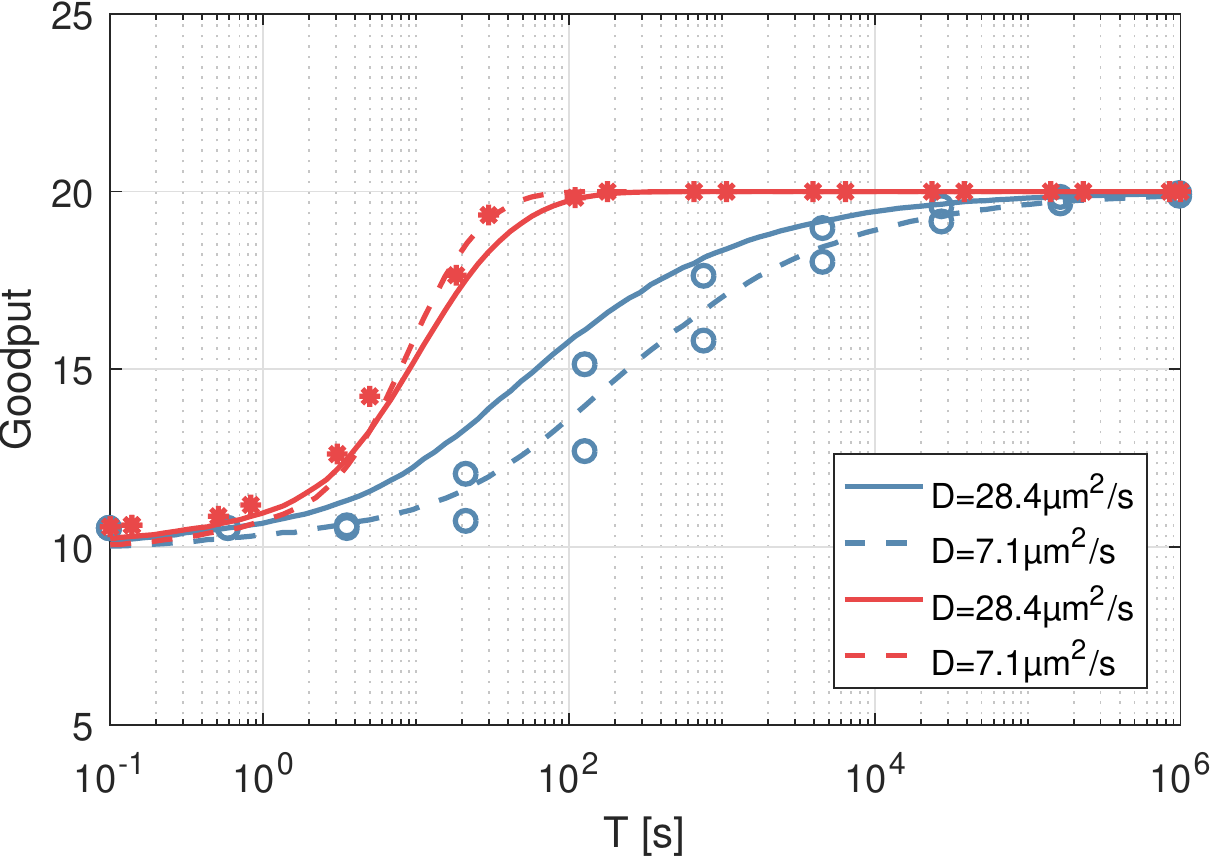}
  \end{center} 
  \vspace{-2ex}
  \caption{Goodput versus  bit interval $T$ for different diffusion coefficients $D$ without background noise. Blue and red curves: Simulated goodput for pure diffusion and drift channels; Circle and star marker: Theoretical goodput approximation for pure diffusion and drift channels as discussed in Sec. \ref{sec:sim_res}.}
  \label{fig:GP_vs_T_D}
\end{figure}

\subsection{Performance without Background Noise}
\label{subsec:perf_wo_bg_noise}
Figs. \ref{fig:BER_vs_T_K} -- \ref{fig:GP_vs_T_D} show the error and goodput performance for different blocks sizes $K$, distances $d$ and  diffusion coefficients~$D$ without background noise, i.e. $\lambda = 0\,\text{s}^{-1}$. 
We observe that the BER decreases as the bit interval $T$ increases, since this lowers the probability of transpositions. The BER for drift channels decreases significantly faster than for pure diffusion channels. This is because in drift channels single transpositions between neighboring bits are dominant for large bit intervals~$T$~\mbox{(cf. Fig. \ref{fig:PermProb_vs_T})}. Moreover, for drift channels the goodput reaches its maximum value $G=K$, i.e. all $K$ released bits are received successfully, already for smaller bit intervals compared to pure diffusion channels. Although, the BER decreases and the goodput increases as the bit interval increases, also the transmission time increases with an increasing bit interval. Thus, the bit interval allows a tradeoff between reliability and transmission time. 

Figs. \ref{fig:BER_vs_T_K} and \ref{fig:GP_vs_T_K} show the BER and goodput performance versus the bit interval~$T$ for various block sizes $K$. We observe from Fig. \ref{fig:BER_vs_T_K} that increasing the block size results in a slight loss in the BER performance for drift channels, but in a significant performance degradation for pure diffusion channels. This is because in drift channels the error performance is dominated
by single transpositions irrespective of the block size and for pure diffusion channels more transpositions occur as the block size increases. A similar observation can be made for the goodput in Fig. \ref{fig:GP_vs_T_K}. For drift channels the bit interval at which the goodput reaches its maximum value $G=K$ is almost independent of the block size, but for pure diffusion channels it increases as the block size increases.

Figs. \ref{fig:BER_vs_T_d} and \ref{fig:GP_vs_T_d} show the BER and goodput performance versus the bit interval~$T$ for various distances $d$ between TX and RX. We observe that the BER decreases as the distance decreases. 
Moreover, for small distances the goodput is higher for a certain bit interval and the maximum value $G=K$ is reached faster.

Figs. \ref{fig:BER_vs_T_D} and \ref{fig:GP_vs_T_D} show the BER and goodput performance versus the bit interval~$T$ for $D\in\{7.1,28.4\}\mu\text{m}^2\text{/s}$. The diffusion coefficient $D = 7.1\,\mu\text{m}^2\text{/s}$ is obtained by increasing the molecule radius from $r=2\,\text{nm}$ to $r=8\,\text{nm}$. Interestingly, we observe from Fig. \ref{fig:BER_vs_T_D} that for drift channels the BER decreases if the radius of the released molecules
is increased (lower diffusion coefficient) \cite{Furubayashi16}. This is because larger molecules experience more support due to drift compared to small molecules. However, for pure diffusion channels the performance decreases when the molecule radius is increased. Hence, for drift channels the goodput in a certain bit interval range can be improved by releasing larger molecules, but for pure diffusion channels this results in a degradation of the goodput.

In Figs. \ref{fig:BER_vs_T_K}, \ref{fig:BER_vs_T_d} and \ref{fig:BER_vs_T_D} we also compare the theoretical BEP approximations in \eqref{eq:bep_appr_pure} and \eqref{eq:bep_appr_drift} with the simulated BER. We show that for drift channels they match very well and we observe that the upper bound derived for pure diffusion channels is tight.
Similar observations can be made for the goodput peformance shown in Figs. \ref{fig:GP_vs_T_K}, \ref{fig:GP_vs_T_d} and~\ref{fig:GP_vs_T_D}, where the theoretical goodput approximation is obtained using~\eqref{eq:goodput}, with $P^\text{e} \approx P_\text{t}^\text{e}$ and using the BEP approximations in \eqref{eq:bep_appr_pure} and \eqref{eq:bep_appr_drift} for $P_\text{t}^\text{e}$.

\begin{figure}[t!]
  \begin{center}
    \includegraphics[scale = 0.7]{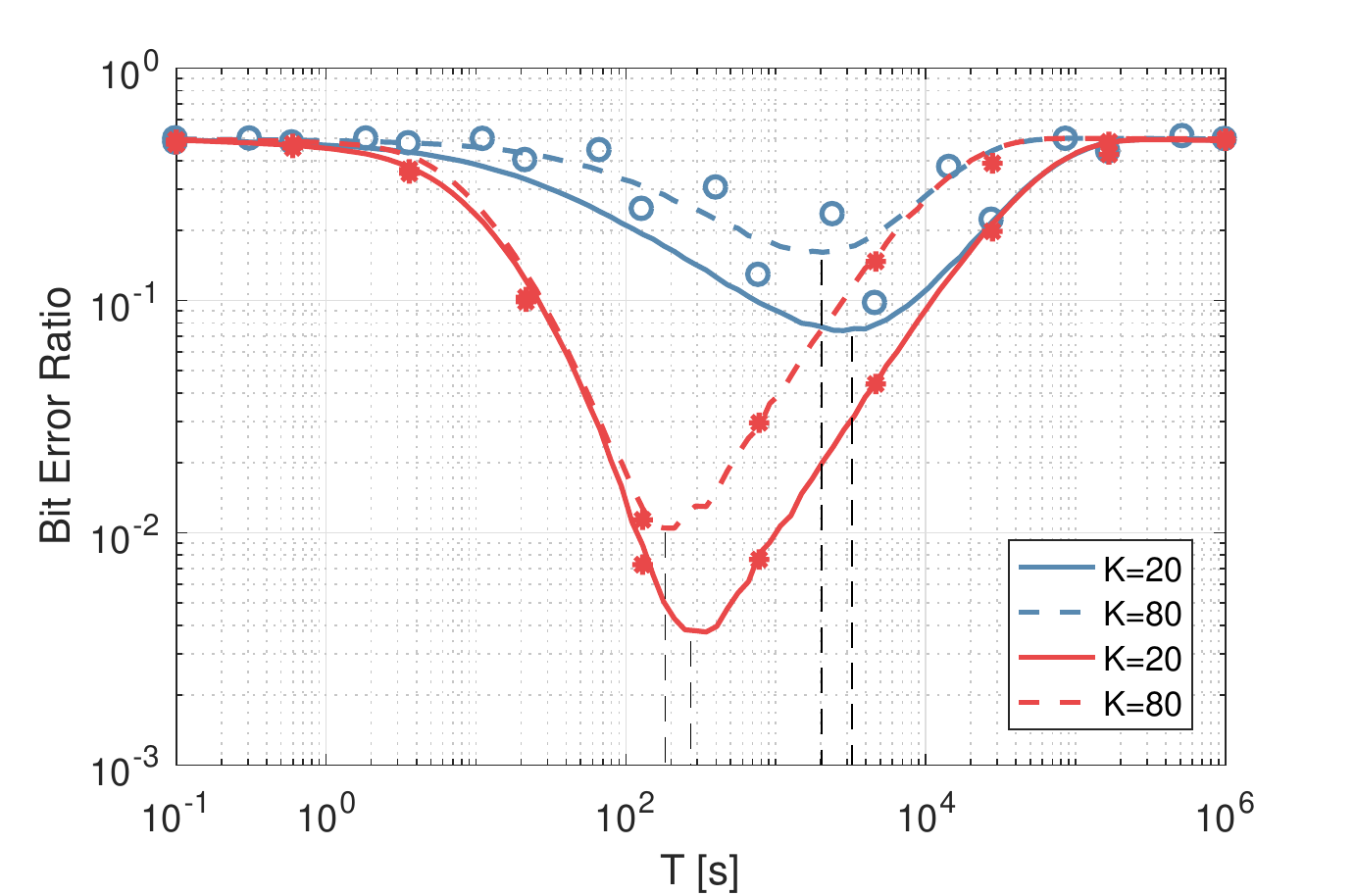}
  \end{center} 
  \vspace{-2ex}
  \caption{BER versus bit interval $T$ for different block sizes~$K$ with noise generation rate $\lambda = 1\times 10^{-6}\,s^{-1}$. Blue and red curves: Simulated BER for pure diffusion and drift channels; Circle and star marker: Theoretical BEP approximation for pure diffusion and drift channels as discussed in Secs. \ref{subsec:appr_bep} and \ref{subsec:bep_bg_noise}.The optimal bit interval~$T_\text{opt}$ derived from \eqref{eq:opt_bit_interval} for pure diffusion and drift channels is given by $\{183,271\}\,s$ and \{2022, 3212\}\,s, respectively.}
  \label{fig:BER_vs_T_K_BG}
\end{figure}

\begin{figure}[t!]
  \begin{center}
    \includegraphics[scale = 0.7]{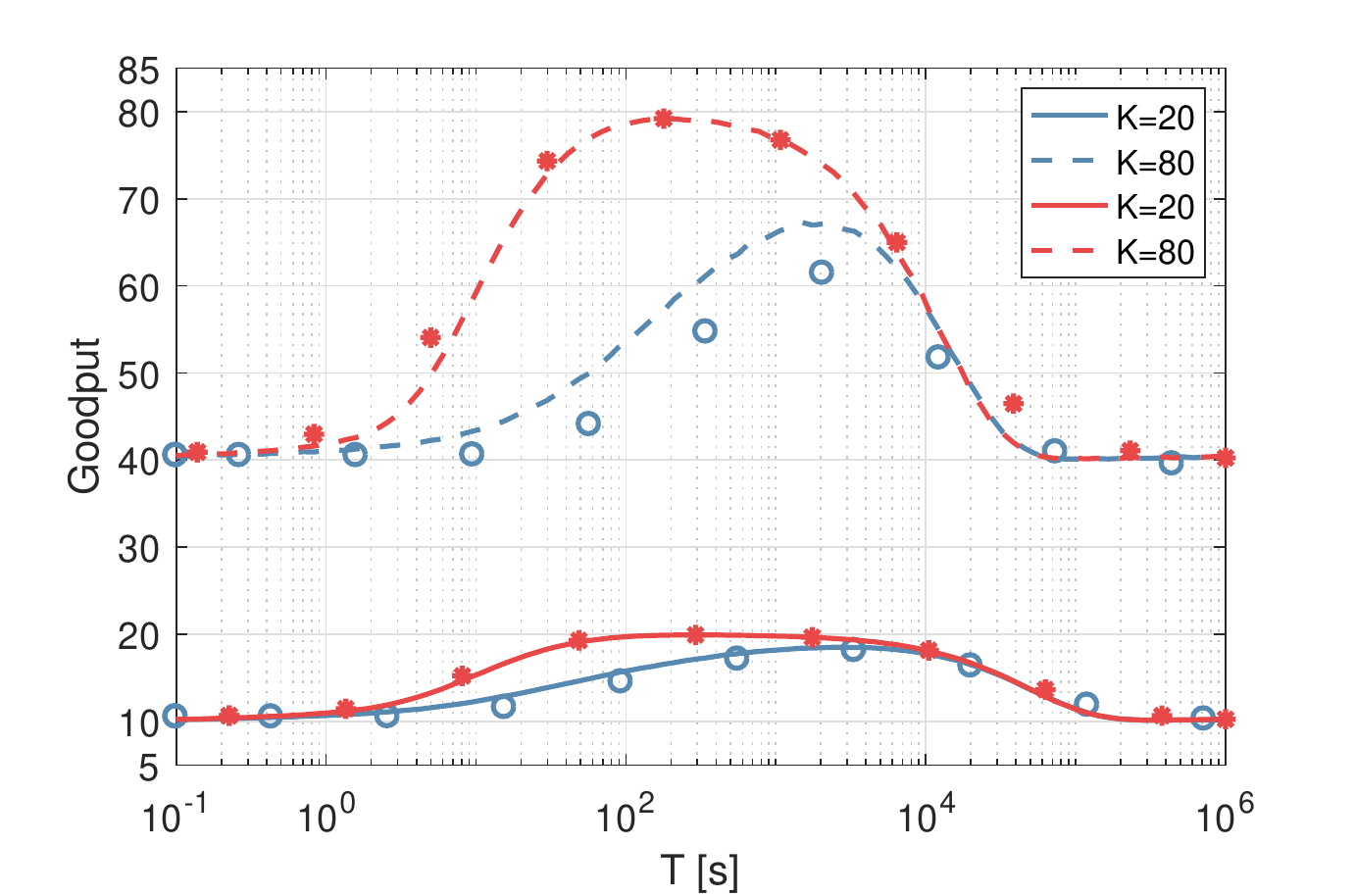}
  \end{center} 
  \vspace{-2ex}
  \caption{Goodput versus bit interval $T$ for different block sizes~$K$ with noise generation rate $\lambda = 1\times 10^{-6}\,s^{-1}$. Blue and red curves: Simulated goodput for pure diffusion and drift channels; Circle and star marker: Theoretical goodput approximation for pure diffusion and drift channels as discussed in Sec. \ref{sec:sim_res}.}
  \label{fig:TP_vs_T_K_BG}
\end{figure}

\begin{figure}[t!]
  \begin{center}
    \includegraphics[scale = 0.7]{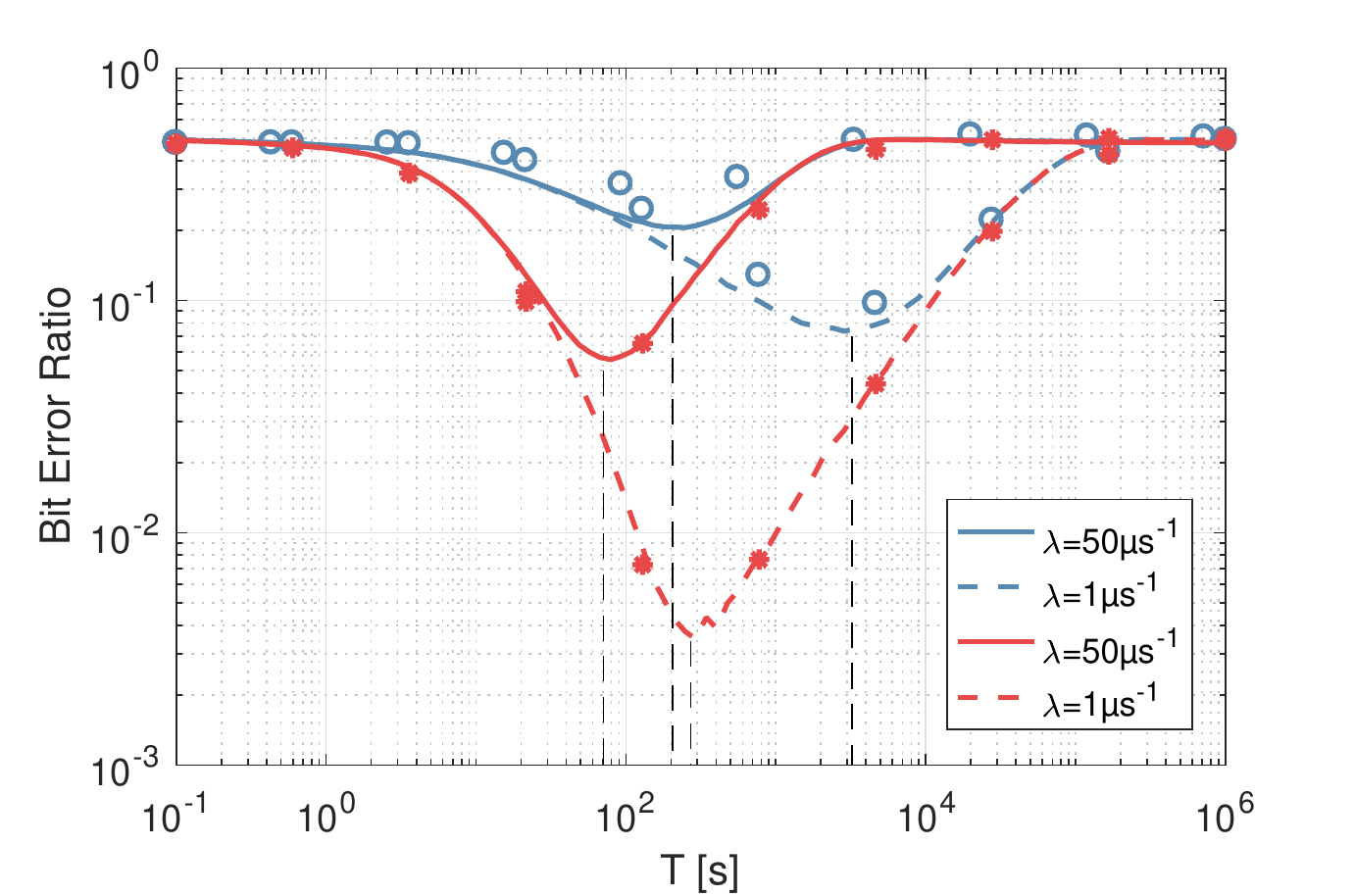}
  \end{center} 
  \vspace{-2ex}
  \caption{BER versus bit interval $T$ for different noise generation rates~$\lambda$. Blue and red curves: Simulated BER for pure diffusion and drift channels; Circle and star marker: Theoretical BEP approximation for pure diffusion and drift channels as defined in Secs. \ref{subsec:appr_bep} and \ref{subsec:bep_bg_noise}. The optimal bit interval~$T_\text{opt}$ derived from \eqref{eq:opt_bit_interval} for pure diffusion and drift channels is given by $\{71,271\}\,s$ and \{204, 3212\}\,s, respectively.}
  \label{fig:BER_vs_T_Lambda}
\end{figure}

\begin{figure}[t!]
  \begin{center}
    \includegraphics[scale = 0.7]{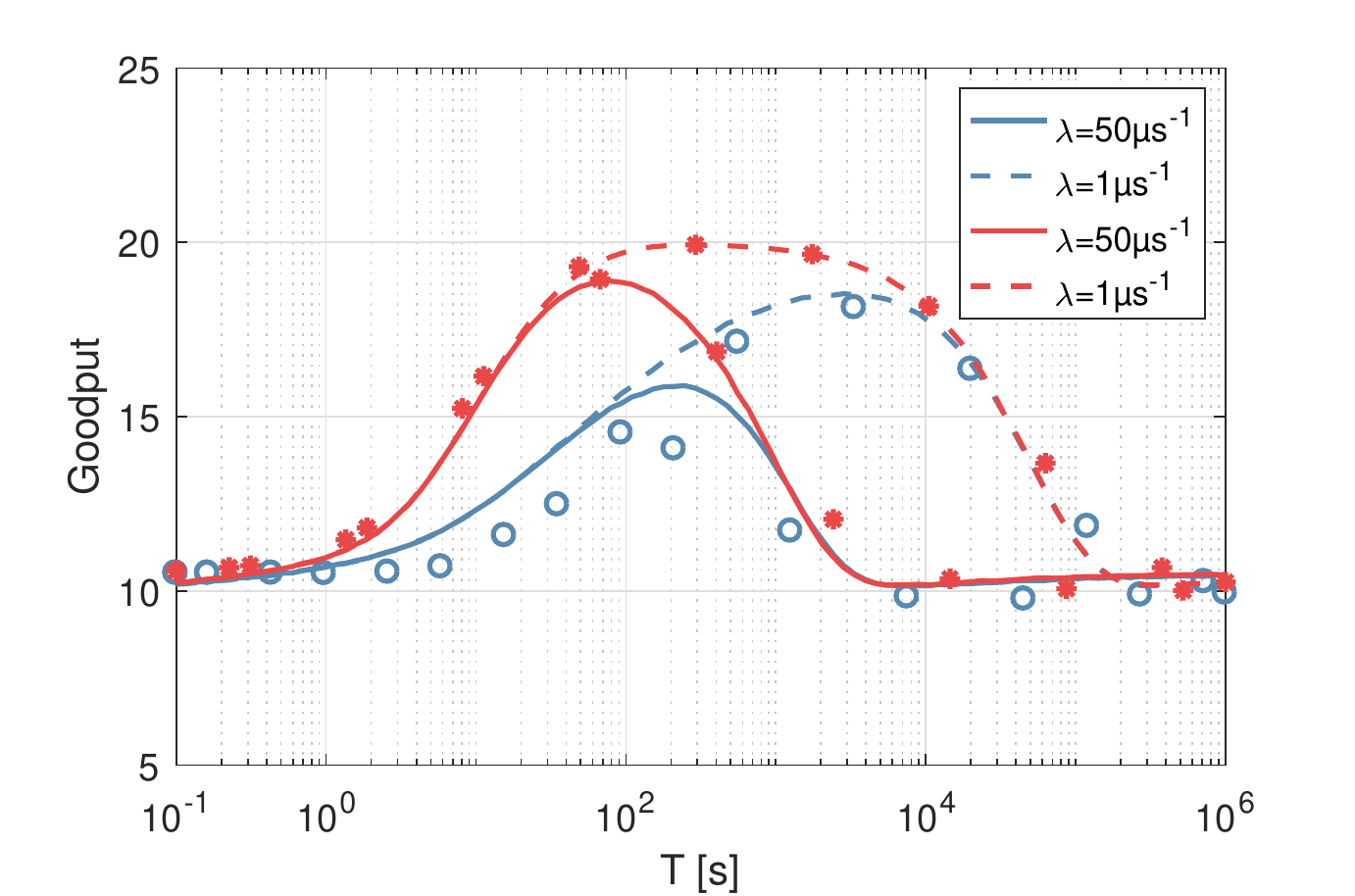}
  \end{center} 
  \vspace{-2ex}
  \caption{Goodput versus bit interval $T$ for different noise generation rates~$\lambda$. Blue and red curves: Simulated BER for pure diffusion and drift channels; Circle and star marker: Theoretical BEP approximation for pure diffusion and drift channels as defined in Sec. \ref{sec:sim_res}.}
  \label{fig:TP_vs_T_Lambda}
\end{figure}

\revMath{
\subsection{Performance with Background Noise}
\label{subsec:perf_bg_noise}
Figs. \ref{fig:BER_vs_T_K_BG} -- \ref{fig:TP_vs_T_Lambda} show the error and goodput performance when background noise is considered. Moreover, we compare the theoretical BEP approximation derived in Secs. \ref{subsec:appr_bep} and \ref{subsec:bep_bg_noise} and the theoretical goodput performance defined in~\eqref{eq:goodput} with the simulation results. We observe a good match between the theoretical and simulation results. In addition, the optimal bit interval ~$T_\text{opt}$ derived from~\eqref{eq:opt_bit_interval} is depicted in Figs.~\ref{fig:BER_vs_T_K_BG}~and~\ref{fig:BER_vs_T_Lambda}, showing a good match between the analytical and the numerical results.

Figs.~\ref{fig:BER_vs_T_K_BG} and \ref{fig:TP_vs_T_K_BG} show the BER and goodput performance versus the bit interval~$T$ for various block sizes $K$, considering background noise with rate $\lambda = 1\times 10^{-6}\,s^{-1}$. We observe that up to the optimal bit interval the BER decreases with the same slope as in the case without background noise~\mbox{(cf. Fig. \ref{fig:BER_vs_T_K})}, i.e. errors due to transpositions are dominant. After the optimal bit interval the BER increases again, i.e. errors due to background noise become dominant. Similarly, the goodput increases until the optimal bit interval and then decreases again. It is important to note that in contrast to the scenario without background noise the goodput does not necessarily achieve its maximum value $G=K$. 
Figs. \ref{fig:BER_vs_T_Lambda} and \ref{fig:TP_vs_T_Lambda} show the BER and goodput performance versus the bit interval~$T$ for various noise generation rates $\lambda$. We observe that if the noise generation rate $\lambda$ is increased the optimal bit interval decreases and, thus, the minimum achievable BER is reduced. Similarly, the maximum achievable goodput becomes lower as the noise generation rate increases.}

%% file: Conclusions.tex

In this work, we comprehensively studied the impact of transposition errors on the performance of diffusion-based MC with and without drift. We used type-based information encoding, releasing a single molecule per information bit and the detection algorithm exploited the arrival order of the molecules. We presented an analytical expression for the exact BEP and derived computationally tractable approximations of the BEP for pure diffusion and drift channels. 
In addition, we also considered the impact of background noise and derived an approximation of the BEP and the optimal bit interval that minimizes the BEP. Numerical results revealed a huge performance gain for drift channels compared to pure diffusion channels in terms of error rate and goodput. We observed that pure diffusion-based MC is only appropriate for short block sizes and if the transmission time is not crucial. Extending the presented results for information encoding using more than two types of molecules would be an interesting future work. This would significantly reduce the error rate and the transmission time.

\balance